\documentclass[conference,letterpaper]{IEEEtran}
\usepackage{balance}
\usepackage{amssymb,amsthm,amsmath,epsfig,latexsym,graphicx,bm,nccmath}
\allowdisplaybreaks[4] 
\usepackage{mathtools}

\usepackage{epstopdf}

\usepackage{booktabs}
\usepackage{tabularx}

\usepackage{enumerate}
\usepackage{cite}
\usepackage{amsfonts}
\usepackage[ruled, vlined, linesnumbered, commentsnumbered, longend]{algorithm2e}
\usepackage{textcomp}
\usepackage{xcolor}
\def\BibTeX{{\rm B\kern-.05em{\sc i\kern-.025em b}\kern-.08em
		T\kern-.1667em\lower.7ex\hbox{E}\kern-.125emX}}
\usepackage[utf8x]{inputenc}
\usepackage{etoolbox}
\usepackage{multirow}

\newtheorem{proposition}{Proposition}

\usepackage[skip=0pt]{caption}
\usepackage[font=footnotesize,skip=0pt]{subcaption}
\usepackage{easyReview} 

\begin{document}
	\bstctlcite{IEEEexample:BSTcontrol}
	\title{Profit-Driven Pricing and SLA-Aware Reserve Sizing for Multi-Tenant Satellite O-RAN Services}
	
	\author{
		\IEEEauthorblockN{Manobendu Sarker\IEEEauthorrefmark{1},  Gunes Karabulut Kurt\IEEEauthorrefmark{1}, and Wael Jaafar\IEEEauthorrefmark{2}} 
		\IEEEauthorblockA{\IEEEauthorrefmark{1}{Poly-Grames Research Center, Department of Electrical Engineering}, Polytechnique Montr\'{e}al,  Canada}
		
		\IEEEauthorblockA{\IEEEauthorrefmark{2}{Department of Software and IT Engineering}, \'{E}cole de Technologie Sup\'{e}rieure, Canada}
		
		\IEEEauthorblockA{
			\IEEEauthorrefmark{1}\{manobendu.sarker, gunes.kurt\}@polymtl.ca, \IEEEauthorrefmark{2}wael.jaafar@etsmtl.ca
		}
	}
	
	\maketitle
	
	\begin{abstract}
		This paper studies a multi-tenant resource allocation problem in a satellite open radio access network (O-RAN) wholesale setting, where heterogeneous traffic classes share a time-varying limited satellite capacity over a finite planning horizon. We formulate a joint pricing and reserve allocation problem from a service provider perspective, where tenant-specific demand exhibits price elasticity and stochastic service requirements subject to strict service-level agreement (SLA) constraints, leading to a coupled economic and reliability-driven bottleneck. A deterministic reformulation is adopted to approximate probabilistic SLA requirements through tractable margin constraints, enabling coordinated control of horizon-wide contract prices and time-varying reserves. The resulting problem is non-convex due to interdependent decisions across tenants, time windows, and service classes. To address this, an alternating optimization (AO) scheme is developed separating pricing and allocation decisions while preserving feasibility and SLA guarantees. Numerical results show that the proposed method achieves near-optimal profit within approximately $1\%$ of a global benchmark, while reducing runtime by up to $22\times$. In contrast, considered baseline schemes incur profit losses exceeding $15\%$ or fail to satisfy SLA constraints. The proposed approach consistently maintains non-positive empirical SLA gaps and achieves up to $30\%$ higher resource utilization than a price-optimization baseline without adaptive reserve control. These results demonstrate that joint economic and resource control enables the provider to efficiently exploit scarce satellite network capacity with reliable service delivery and scalable computation.
	\end{abstract}
	
	\begingroup
	\renewcommand\thefootnote{}\footnote{
		The authors would like to thank the Institut de valorisation des donn\'{e}es (IVADO) for their financial support.
	}
	\addtocounter{footnote}{-1}
	\endgroup
	\vspace{-6mm}
	\section{Introduction}
	
	The evolution of sixth-generation (6G) networks relies on the integration of low Earth orbit (LEO) and very low Earth orbit (VLEO) constellations to complement terrestrial infrastructure and ensure global, ubiquitous coverage \cite{Lee2025, Benchoubane2025}. Within this framework, a multi-tenant satellite open radio access network (O-RAN) wholesale model enables a satellite infrastructure provider (SInP) to lease capacity to multiple telecom operators (TOs) \cite{Baena2025, Lee2025}. TOs support heterogeneous service portfolios composed mainly of enhanced mobile broadband (eMBB) and ultra-reliable low latency communication (URLLC) traffic, each characterized by distinct service-level agreement (SLA) requirements \cite{Motalleb2023, Raftopoulos2024}. A fundamental tension arises because under-reserving resources increases SLA violation risk and associated penalties, whereas over-reserving resources reduces the SInP's profitability and leads to inefficient capacity utilization \cite{Habibi2018, Chiang2014}. This raises a key question: \textit{how should a satellite O-RAN provider jointly set tenant-specific reservation prices and allocate capacity to satisfy SLA constraints while maximizing profit?}
	
	Addressing this question requires a coordinated treatment of pricing and resource allocation, which is not provided by existing approaches. Indeed, existing work addressing pricing and resource management in integrated networks can be broadly classified into four categories. First, technical slicing solutions for non-terrestrial networks (NTNs) provide robust allocation mechanisms but do not incorporate wholesale pricing \cite{He2024, Lee2025, Benchoubane2025}. Second, market-driven auction models determine prices through strategic bidding but lack the stability required for long-term wholesale contracts \cite{Dieye2020,Jiang2023}. Third, terrestrial slicing frameworks optimize resource orchestration without accounting for NTN-specific constraints and mixed service portfolios \cite{Sciancalepore2017, Habibi2018}. Finally, current O-RAN architectures focus on real-time radio resource control rather than provider-side economic coordination \cite{Raftopoulos2024}. Consequently, none of these approaches jointly address contract-based pricing and reserve sizing under SLA risk and limited satellite capacity.
	
	Building on this, we formulate a joint pricing and reserve allocation problem from the perspective of the SInP, capturing price-dependent demand and stochastic service requirements under SLA constraints. To the best of our knowledge, the joint design of contract-based wholesale pricing and time-varying reserve allocation under SLA risk in satellite O-RAN systems has not been investigated. This formulation explicitly captures the coupling between price-dependent demand and reliability constraints, thus enabling a provider-centric perspective where economic decisions directly influence service feasibility and resource utilization. 
	
	To efficiently address the resulting coupled decision structure, we employ an alternating optimization (AO) approach that jointly captures the interaction between economic control and resource allocation. We consider a finite-horizon model in which wholesale prices are fixed over the contract duration, while admitted reserves are adapted across time windows to account for time-varying capacity and stochastic demand. The main contributions of this paper are summarized as follows:
	\begin{enumerate}
		\item We develop a contract-based joint pricing and reserve allocation model for multi-tenant satellite O-RAN systems that captures heterogeneous service portfolios and SLA-driven reliability risk under a shared and time-varying capacity budget.
		
		\item We propose an AO framework that decomposes pricing and reserve allocation into tractable subproblems while preserving their coupling through demand elasticity and SLA constraints, resulting in improved scalability compared to globally coupled optimization methods.
		
		\item We validate the proposed approach through numerical evaluations, demonstrating near-optimal profit performance, consistent SLA satisfaction, and efficient resource utilization across varying capacity regimes, thereby confirming its suitability for SInP-side orchestration under satellite capacity constraints.
	\end{enumerate}
	\vspace{-3mm}
	
	\section{System Model}
	\label{sec:system}
	
	\subsection{Network Architecture and Resource Abstraction}
	We consider a multi-tenant satellite O-RAN wholesale framework in which an
	infrastructure provider, SInP, offers reservation-based satellite services to
	mobile network operators \cite{Lee2025,Khodashenas2016}. The SInP operates a
	centralized business-control layer that abstracts PHY scheduling and
	packet-level coordination, consistent with higher-layer O-RAN and NTN
	orchestration architectures \cite{Baena2025,Campana2023}. The operational
	horizon is divided into time intervals $t\in\mathcal{T}=\{1,2,\ldots,T\}$,
	each of duration $\Delta t$. Time-varying beam visibility and orbital
	dynamics within each window are aggregated into a conservative,
	time-dependent service budget $C_{\mathrm{eff}}^{(t)}$. This budget represents a lower bound on
	available throughput that accounts for sub-window beam transitions and link
	degradation, thereby capturing the limited and intermittent availability of satellites without requiring explicit beam-level or millisecond-level indices
	\cite{Mahyoub2025,Lee2025,He2024}. 
	
	We consider a set $\mathcal{I}=\{1,2,\ldots,I\}$ of TO tenants
	competing for the shared satellite service budget, each with a service
	portfolio that requires a reserve $r_i^{(t)}$ within the satellite service budget, $\forall i \in \mathcal{I}$ and $t \in \mathcal{T}$. The aggregate admitted reserves in every time window $t$ must
	satisfy
	\begin{equation}
		\sum_{i\in\mathcal{I}} r_i^{(t)} \le C_{\mathrm{eff}}^{(t)},
		\qquad \forall t\in\mathcal{T}.
		\label{eq:capacity_operator}
	\end{equation}
	For tractability, the model does not impose inter-window ramping or
	handover-transition costs on $r_i^{(t)}$, thus each planning window is optimized
	using updated window-level service conditions.
	
	\vspace{-2mm}
	\subsection{Physical Layer and Delay Abstraction}
	\label{subsec:phy_abstraction}
	
	In LEO satellite systems, a satellite serves a given terrestrial area only
	during short visibility windows with time-varying propagation conditions,
	including line-of-sight (LoS)/Non-LoS (NLoS) transitions, Doppler shifts, and distance-dependent path
	loss \cite{Mahyoub2025,3GPP_TR38811}. The user--satellite--gateway service
	path also introduces non-negligible propagation delay, particularly when
	processing is located at the gateway rather than onboard
	\cite{Baena2025,Lee2025}.
	
	To reconcile physical-layer dynamics with the planning horizon, end-to-end (E2E)
	propagation delay is not explicitly modeled. Instead, it is captured indirectly
	through stricter SLA conservatism for delay-sensitive services. We assume the baseline service path meets nominal URLLC propagation-delay requirements, so
	the tighter reliability tolerance protects against queueing and
	service-related degradation rather than compensating for propagation physics
	itself. Accordingly, the reliability tolerance assigned to URLLC-like traffic
	is set significantly tighter than that of eMBB-like traffic (as will be discussed in 
	Section~\ref{subsec:operator_portfolios}).
	
	\subsection{Operator Service Portfolios and Demand Uncertainty}
	\label{subsec:operator_portfolios}
	Each TO tenant carries a mixed service portfolio composed of a URLLC-like class
	$u$ and an eMBB-like class $b$, reflecting the heterogeneous requirements of
	O-RAN slicing \cite{Motalleb2023}. Let $\mathcal{K}=\{u,b\}$, and
	the portfolio class weights $\alpha_{i,k}$, $\forall k\in\mathcal{K}$, satisfy
	$\sum_{k\in\mathcal{K}}\alpha_{i,k}=1$, $\forall i \in \mathcal{I}$, and are held constant
	over the operational horizon.
	
	During each time window $t$, the SInP sets a TO-specific wholesale price
	$p_i^{(t)}$, $\forall i \in \mathcal{I}$, which induces a price-sensitive reservation request. We adopt the
	linear price-response model \cite{Chiang2014}
	\begin{equation}
		q_i^{(t)}=\bar{q}_i^{(t)}-\kappa_i p_i^{(t)},
		\qquad \forall i\in\mathcal{I},\; \forall t\in\mathcal{T},
		\label{eq:request_operator}
	\end{equation}
	where $\bar{q}_i^{(t)}>0$ is the baseline request at zero price and
	$\kappa_i>0$ is the tenant-specific price-sensitivity coefficient. This
	captures the intuitive trade-off that higher prices reduce requested volume,
	while remaining simple enough for SInP-side planning
	\cite{Khodashenas2016,Chiang2014}. Then the  admitted tenant-level reserve
	$r_i^{(t)}$ satisfies $0\le r_i^{(t)}\le q_i^{(t)}$,
	$\forall i\in\mathcal{I},\; \forall t\in\mathcal{T}$.
	
	Let $D_{i,k}^{(t)}$ be the realized instantaneous demand of tenant $i$ for
	service class $k$ in window $t$, and
	$D_i^{(t)}=\sum_{k\in\mathcal{K}} D_{i,k}^{(t)}$. We model demand as
	stochastic with independent normal marginals as follows \cite{He2024}:
	\begin{equation}
		D_{i,k}^{(t)}\sim\mathcal{N}\!\big(\mu_{i,k}^{(t)},\sigma_{i,k}^2\big),
		\qquad \forall i\in\mathcal{I},\; \forall k\in\mathcal{K},\; \forall t\in\mathcal{T},
		\label{eq:gaussian_operator}
	\end{equation}
	where
	\begin{equation}
		\mu_{i,k}^{(t)}=\beta_{i,k}\alpha_{i,k}q_i^{(t)},
		\qquad 0<\beta_{i,k}\le 1,
		\label{eq:mean_operator}
	\end{equation}
	and $\beta_{i,k}$ is the mean instantaneous consumption ratio relative to the
	requested envelope \cite{Habibi2018,Chiang2014}. We assume
	$\mu_{i,k}^{(t)}/\sigma_{i,k}$ is sufficiently large that the probability of
	negative demand realizations is negligible in the considered operating regime.
	
	The class-specific reserve available to tenant $i$ in window $t$ is
	approximated as $\alpha_{i,k}r_i^{(t)}$, i.e., the SInP controls only
	$r_i^{(t)}$ while the intra-tenant split follows the declared portfolio. An
	SLA violation for class $k$ occurs when
	$D_{i,k}^{(t)}>\alpha_{i,k}r_i^{(t)}$, and the violation probability is
	\begin{equation}
		P_{i,k}^{(t)} \triangleq
		\Pr\!\big(D_{i,k}^{(t)}>\alpha_{i,k}r_i^{(t)}\big),
		\ \forall i\in\mathcal{I},\; \forall k\in\mathcal{K},\; \forall t\in\mathcal{T}.
		\label{eq:sla_operator}
	\end{equation}
	The system enforces class-specific reliability targets
	\begin{equation}
		P_{i,k}^{(t)} \le \varepsilon_{i,k},
		\qquad \forall i\in\mathcal{I},\; \forall k\in\mathcal{K},\; \forall t\in\mathcal{T},
		\label{eq:sla_target}
	\end{equation}
	where $\varepsilon_{i,u}\ll\varepsilon_{i,b}$ reflects the stricter
	reliability requirement of URLLC compared to eMBB
	\cite{Motalleb2023,Raftopoulos2024}.
	
	\subsection{Economic Framework and Inter-Tenant Coupling}
	\label{subsec:economic_operator}
	The SInP trades reservation revenue against reserve-maintenance cost and
	expected shortfall penalties, consistent with SLA-aware profit models
	\cite{Khodashenas2016,Chiang2014,Habibi2018}. Let $c$ be the unit cost of
	maintaining reserved capacity per window, $w_i$ the tenant-specific penalty
	coefficient, and $\underline r_i$ the minimum per-tenant reservation floor
	that prevents starvation. To capture the higher criticality of URLLC
	shortfalls, we introduce class-specific severity weights $\gamma_k$ with
	$\gamma_u>\gamma_b$, and define the portfolio shortfall of tenant $i$ in
	window $t$ as
	\begin{equation}
		S_i^{(t)}=\sum_{k\in\mathcal{K}} \gamma_k
		\max\{0,\; D_{i,k}^{(t)}-\alpha_{i,k}r_i^{(t)}\}.
		\label{eq:shortfall_operator}
	\end{equation}
	The reliability constraint \eqref{eq:sla_target} and the expected shortfall
	$\mathbb{E}[S_i^{(t)}]$ jointly capture the operational risk, where the latter
	quantifies the expected monetary loss from unmet demands
	\cite{Chiang2014,Habibi2018}.
	
	Since TO tenants share the service budget $C_{\mathrm{eff}}^{(t)}$ in each
	window, increasing the reserve of one TO reduces the capacity available
	to others \cite{He2024,Dieye2020}. This shared-capacity coupling creates
	a competitive environment that requires joint optimization of pricing and
	reserve sizing over the horizon to maintain SInP profitability and
	inter-tenant fairness \cite{Lee2025,Jiang2023,Khodashenas2016}.
	\vspace{-2mm}
	\section{Problem Formulation}
	\label{sec:problem}
	
	We assume that the system parameters are chosen such that the feasible set
	of the problem defined below is nonempty, i.e., the per-tenant SLA margins
	and service floors are jointly satisfiable within the satellite capacity
	budget over the operating horizon.\footnote{If a given parameter
		combination renders the problem infeasible, the provider must relax either
		the capacity target or the reliability constraints before accepting the
		slice requests.}
	Although the system model admits window-dependent price signals
	\(p_i^{(t)}\), we adopt a horizon-wide wholesale contract price \(p_i\) for
	each tenant to avoid unrealistic price volatility across planning windows,
	while allowing the admitted reserve \(r_i^{(t)}\) to vary over
	time.\footnote{The formulation is therefore window-coupled only through the
		horizon-wide prices \(\{p_i\}_{i\in\mathcal{I}}\). No ramping or switching
		cost is imposed on the reserve sequence
		\(\{r_i^{(t)}\}_{t\in\mathcal{T}}\) in the present paper.}
	Under this abstraction the SInP determines the operator-specific prices
	\(\{p_i\}_{i\in\mathcal{I}}\) and the admitted reserves
	\(\{r_i^{(t)}\}_{i\in\mathcal{I},\,t\in\mathcal{T}}\) before the realization
	of the stochastic class-specific demands
	\(\{D_{i,k}^{(t)}\}_{i\in\mathcal{I},\,k\in\mathcal{K},\,t\in\mathcal{T}}\)
	\cite{Chiang2014}. The SInP seeks to maximize the expected profit from
	reservation revenue, net of reserve-maintenance cost and expected SLA
	shortfall loss, while respecting the time-dependent capacity budget,
	request limits, reliability targets, and a minimum per-tenant service floor.
	The resulting planning problem is formulated as
	\begin{subequations}
		\label{prob:joint_pricing_reserve_operator}
		\begin{align}
			\max_{\{p_i,r_i^{(t)}\}} \quad
			& \sum_{t\in\mathcal{T}} \sum_{i\in\mathcal{I}}
			\Big(
			p_i r_i^{(t)} - c\,r_i^{(t)} - w_i\,\mathbb{E}[S_i^{(t)}]
			\Big)
			\label{eq:obj_joint_pricing_reserve_operator}
			\\
			\text{s.t.} \quad
			& \sum_{i\in\mathcal{I}} r_i^{(t)} \le C_{\mathrm{eff}}^{(t)},
			\qquad \forall t\in\mathcal{T},
			\label{eq:cap_joint_pricing_reserve_operator}
			\\
			& 0 \le p_i \le \min_{t\in\mathcal T} \frac{\bar{q}_i^{(t)}}{\kappa_i},
			\qquad \forall i\in\mathcal{I},
			\label{eq:price_joint_pricing_reserve_operator}
			\\
			& 0 \le r_i^{(t)} \le \bar{q}_i^{(t)}-\kappa_i p_i,
			\qquad \forall i\in\mathcal{I},\; \forall t\in\mathcal{T},
			\label{eq:request_link_joint_pricing_reserve_operator}
			\\
			& P_{i,k}^{(t)} \le \varepsilon_{i,k},
			\qquad \forall i\in\mathcal{I},\; \forall k\in\mathcal{K},\; \forall t\in\mathcal{T},
			\label{eq:chance_joint_pricing_reserve_operator}
			\\
			& r_i^{(t)} \ge \underline r_i,
			\qquad \forall i\in\mathcal{I},\; \forall t\in\mathcal{T}.
			\label{eq:floor_joint_pricing_reserve_operator}
		\end{align}
	\end{subequations}
	In~\eqref{eq:obj_joint_pricing_reserve_operator}, the first term is the
	reservation revenue collected from operator tenant \(i\) during window \(t\),
	the second term is the cost of maintaining the admitted reserve, and the
	third term represents the expected financial loss associated with the
	portfolio shortfall \(S_i^{(t)}\) defined in~\eqref{eq:shortfall_operator}. Constraint
	\eqref{eq:cap_joint_pricing_reserve_operator} imposes the time-dependent
	satellite-capacity limit. Constraint
	\eqref{eq:price_joint_pricing_reserve_operator} preserves nonnegative
	reservation requests under the linear price-response model, while
	\eqref{eq:request_link_joint_pricing_reserve_operator} ensures that the
	SInP does not admit more reserves than what is requested at the offered
	price in each window. Constraint
	\eqref{eq:chance_joint_pricing_reserve_operator} enforces the class-specific
	reliability target within each tenant portfolio, with stricter tolerances for URLLC-like traffic than for eMBB-like traffic. Finally, constraint
	\eqref{eq:floor_joint_pricing_reserve_operator} introduces a minimum
	per-tenant service floor $\underline r_i$ to prevent starvation under scarcity.
	
	Problem~\eqref{prob:joint_pricing_reserve_operator} is a finite-horizon
	chance-constrained stochastic nonlinear program. The stochasticity enters
	through the realized class-specific demands across the planning windows,
	while the formulation is generally nonconvex due to the bilinear revenue term
	\(p_i r_i^{(t)}\) and the nonlinear structure induced by the shortfall
	term and the reliability constraints. We next derive a deterministic equivalent of the chance constraints under the Gaussian demand model.
	
	\subsection{Deterministic Reformulation}
	\label{subsec:deterministic}
	
	For each tenant \(i\in\mathcal{I}\), service
	class \(k\in\mathcal{K}\), and planning window \(t\in\mathcal{T}\), the
	reliability constraint \eqref{eq:chance_joint_pricing_reserve_operator} reduces under the Gaussian demand model~\eqref{eq:gaussian_operator} as
	\begin{equation}
		\alpha_{i,k} r_i^{(t)} \ge \mu_{i,k}^{(t)} +
		z_{1-\varepsilon_{i,k}}\sigma_{i,k},
		\qquad \forall i\in\mathcal{I},\; \forall k\in\mathcal{K},\; \forall t\in\mathcal{T},
		\label{eq:det_reliability}
	\end{equation}
	where \(z_{1-\varepsilon_{i,k}}\) is the \((1-\varepsilon_{i,k})\)-quantile of
	the standard normal distribution. Consequently, the expected shortfall admits a closed form. Specifically, define
	\begin{equation}
		\zeta_{i,k}^{(t)} \triangleq
		\frac{\alpha_{i,k} r_i^{(t)}-\mu_{i,k}^{(t)}}{\sigma_{i,k}},
		\qquad \forall i\in\mathcal{I},\; \forall k\in\mathcal{K},\; \forall t\in\mathcal{T},
		\label{eq:zeta_ik}
	\end{equation}
	and let \(\phi(\cdot)\) and \(\Phi(\cdot)\) denote the standard normal density
	and cumulative distribution function. The expected class-level shortfall is
	\begin{align}
		& \mathbb{E}\!\left[\max\{0,D_{i,k}^{(t)}-\alpha_{i,k}r_i^{(t)}\}\right]
		=
		\sigma_{i,k}\phi(\zeta_{i,k}^{(t)}) \nonumber\\
		& \qquad \qquad \qquad +
		\big(\mu_{i,k}^{(t)}-\alpha_{i,k}r_i^{(t)}\big)
		\big[1-\Phi(\zeta_{i,k}^{(t)})\big],
		\label{eq:class_shortfall_closed}
	\end{align}
	and the portfolio shortfall expectation follows as
	\begin{equation}
		\mathbb{E}[S_i^{(t)}]
		= \sum_{k\in\mathcal{K}}
		\gamma_k
		\Big[
		\sigma_{i,k}\phi(\zeta_{i,k}^{(t)})
		+ \big(\mu_{i,k}^{(t)}-\alpha_{i,k}r_i^{(t)}\big)
		\big(1-\Phi(\zeta_{i,k}^{(t)})\big)
		\Big].
		\label{eq:tenant_shortfall_closed}
	\end{equation}
	
	Substituting~\eqref{eq:det_reliability} and~\eqref{eq:tenant_shortfall_closed}
	into the stochastic program yields a deterministic nonlinear program whose
	feasible set is linear. Hence, the transformed planning problem from~\eqref{prob:joint_pricing_reserve_operator} reads
	\begin{subequations}
		\label{prob:deterministic_operator}
		\begin{align}
			\max_{\{p_i,\,r_i^{(t)}\}} \
			& \sum_{t\in\mathcal{T}} \sum_{i\in\mathcal{I}}
			\left(
			p_i r_i^{(t)} - c\,r_i^{(t)} - w_i\,\mathbb{E}[S_i^{(t)}]
			\right)
			\label{eq:obj_det_operator}
			\\
			\text{s.t.} \quad
			& \text{\eqref{eq:cap_joint_pricing_reserve_operator},
				\eqref{eq:price_joint_pricing_reserve_operator},
				\eqref{eq:request_link_joint_pricing_reserve_operator},
				\eqref{eq:floor_joint_pricing_reserve_operator}, and \eqref{eq:det_reliability}.} 
		\end{align}
	\end{subequations}
	Problem~\eqref{prob:deterministic_operator} possesses a linear
	feasible region but a nonconvex objective due to the bilinear term
	\(p_i r_i^{(t)}\) and the nonlinear shortfall components. The next section
	develops a specialized solution algorithm that handles this nonconvexity.
	\vspace{-2mm}
	\section{Proposed Solution}
	\label{sec:proposed_solution}
	
	Problem~\eqref{prob:deterministic_operator} is the deterministic equivalent
	of the original stochastic program under the Gaussian demand model. We solve
	it by an AO scheme that isolates the few
	horizon-wide contract prices $\{p_i\}_{i\in\mathcal I}$ from the per-window
	reserves $\{r_i^{(t)}\}_{i\in\mathcal I,t\in\mathcal T}$. For fixed prices
	the reserve subproblem admits an exact, efficient solution and the prices are
	then refined by a low-dimensional coordinate search.
	
	For each tenant $i$, define $p_i^{\max}\triangleq\min_{t\in\mathcal T}
	\bar{q}_i^{(t)}/\kappa_i$, the upper price bound implied by
	\eqref{eq:price_joint_pricing_reserve_operator}. With prices fixed, the
	optimal reserves are obtained from the value function
	\begin{equation}
		\mathcal F(\{p_i\})\triangleq\;
		\max_{\{r_i^{(t)}\}}\;
		\sum_{t\in\mathcal T}\sum_{i\in\mathcal I}
		\Bigl(p_i r_i^{(t)}-c r_i^{(t)}-w_i\,\mathbb E[S_i^{(t)}]\Bigr),
		\label{eq:value_function}
	\end{equation}
	subject to \eqref{eq:cap_joint_pricing_reserve_operator},
	\eqref{eq:request_link_joint_pricing_reserve_operator},
	\eqref{eq:floor_joint_pricing_reserve_operator}, and
	\eqref{eq:det_reliability}. The outer problem is then
	\begin{equation}
		\max_{\{p_i\}}\;\mathcal F(\{p_i\})
		\quad\text{s.t.}\quad
		\eqref{eq:price_joint_pricing_reserve_operator},
		\label{eq:outer_problem}
	\end{equation}
	which is a low-dimensional nonlinear program. 
	
	For fixed prices, the inner objective is concave in
	$\{r_i^{(t)}\}$ as the revenue term is linear and the expected
	shortfall $\mathbb E[S_i^{(t)}]$ is convex in the reserves. Hence, problem
	\eqref{eq:value_function} decouples across time windows, and for
	each $t\in\mathcal T$, we solve
	\begin{subequations}
		\label{prob:inner_window}
		\begin{align}
			\max_{\{r_i^{(t)}\}_{i\in\mathcal I}} \quad
			& \sum_{i\in\mathcal I}
			\Bigl((p_i-c)r_i^{(t)}-w_i\,\mathbb E[S_i^{(t)}]\Bigr)
			\label{eq:inner_window_obj}
			\\
			\text{s.t.}\quad
			& \eqref{eq:cap_joint_pricing_reserve_operator},
			\qquad
			L_i^{(t)}\le r_i^{(t)}\le \bar q_i^{(t)}-\kappa_i p_i,\;\forall i\in\mathcal I,
			\label{eq:inner_window_box}
		\end{align}
	\end{subequations}
	where, for given prices, the lower bound
	\begin{equation}
		L_i^{(t)}\triangleq\max\Bigl\{
		\underline r_i,\;
		\max_{k\in\mathcal K}
		\Bigl(\beta_{i,k}(\bar q_i^{(t)}-\kappa_i p_i)
		+\frac{z_{1-\varepsilon_{i,k}}\sigma_{i,k}}{\alpha_{i,k}}\Bigr)
		\Bigr\}
		\label{eq:lower_bound_window}
	\end{equation}
	consolidates the minimum service floor and the class-specific SLA margins. The term $\beta_{i,k}(\bar q_i^{(t)}-\kappa_i p_i)$ in \eqref{eq:lower_bound_window} follows from substituting \eqref{eq:request_operator} into \eqref{eq:mean_operator} and rewriting \eqref{eq:det_reliability} as a lower bound on $r_i^{(t)}$. For a fixed price vector, feasibility of the inner reserve-allocation problem requires $L_i^{(t)} \leq \bar q_i^{(t)}-\kappa_i p_i,\quad \forall i,t$ and $\sum_{i\in\mathcal I} L_i^{(t)} \leq C_{\mathrm{eff}}^{(t)},\quad \forall t.$ Candidate price vectors that violate either condition are discarded. The optimality conditions of problem \eqref{prob:inner_window} are expressed through
	the marginal profit
	\begin{equation}
		\psi_i^{(t)}(r)\triangleq
		(p_i-c)+w_i\sum_{k\in\mathcal K}\gamma_k\alpha_{i,k}
		\Bigl[1-\Phi\!\Bigl(
		\frac{\alpha_{i,k}r-\mu_{i,k}^{(t)}}{\sigma_{i,k}}\Bigr)\Bigr],
		\label{eq:psi_window}
	\end{equation}
	which is continuous and nonincreasing in $r$. The second term in $\psi_i^{(t)}(r)$ represents the marginal reduction in the expected SLA shortfall obtained by increasing the admitted reserve. Using the Karush–Kuhn–Tucker (KKT) conditions, the optimal reserve in
	window $t$ satisfies
	\begin{equation}
		r_i^{(t)}=
		\begin{cases}
			L_i^{(t)}, & \psi_i^{(t)}(L_i^{(t)})\le \lambda^{(t)},\\[1mm]
			\bar q_i^{(t)}-\kappa_i p_i, & \psi_i^{(t)}(\bar q_i^{(t)}-\kappa_i p_i)\ge \lambda^{(t)},\\[1mm]
			(\psi_i^{(t)})^{-1}(\lambda^{(t)}), & \text{otherwise},
		\end{cases}
		\label{eq:reserve_characterization}
	\end{equation}
	with a dual variable $\lambda^{(t)}\ge0$ chosen such that
	$\sum_i r_i^{(t)}\le C_{\mathrm{eff}}^{(t)}$ and the complementary slackness
	holds. Since $\psi_i^{(t)}$ is monotone, \eqref{eq:reserve_characterization}
	is solved by bisection on $\lambda^{(t)}$, where each bisection step evaluates
	$\psi_i^{(t)}$ for all $i \in \mathcal{I}$, costing $\mathcal O(|\mathcal I||\mathcal K|)$,
	and the inner-inverse recovery adds a factor $B_r$ representing the bisection
	iteration count for the reserve variable.
	
	The outer problem \eqref{eq:outer_problem} is handled by a deterministic
	coordinate-wise ascent. Specifically, starting from a feasible price vector, one
	coordinate at a time is perturbed within $[0,p_i^{\max}]$, where  the inner problem  for each
	candidate is re‑solved exactly, and only strictly
	improving feasible iterates are accepted. When a full sweep produces no
	improvement, the coordinate step sizes are halved. This procedure is summarized
	in Algo.~\ref{alg:proposed_ao}, which enforces a strict ascent rule
	(Line~\ref{ln:accept}) and shrinks step sizes when the objective becomes
	locally flat (Line~\ref{ln:shrink}), thereby stabilizing convergence near the
	optimum. The inner problem is re‑solved to optimality at every trial price,
	so every accepted iterate is feasible for
	\eqref{prob:deterministic_operator}.
	
	\begin{algorithm}[t]
		\small
		\caption{AO Scheme for Joint Price and Reserve Optimization}
		\label{alg:proposed_ao}
		\DontPrintSemicolon
		\LinesNumbered
		\KwIn{$\bar q_i^{(t)},\kappa_i,\alpha_{i,k},\beta_{i,k},\sigma_{i,k},\varepsilon_{i,k},w_i,\underline r_i$, $\{C_{\mathrm{eff}}^{(t)}\}_{t\in\mathcal T}$, $c$, tolerances $\epsilon_p=10^{-4},\epsilon_{\mathrm{obj}}=10^{-6}$,
			$\Delta_{\min}=10^{-3}$, maximum iterations $M_{\max}=30$, decay factor $\xi=0.5$,
			and price upper bounds $\{p_i^{\max}\}_{i\in\mathcal I}$}
		\KwOut{$\{p_i^\star\}_{i\in\mathcal I}$, $\{r_i^{(t)\star}\}_{i\in\mathcal I,t\in\mathcal T}$}
		Construct a feasible initial price vector $\{p_i^{(0)}\}$ and set initial step sizes $\Delta_i>0$\;
		Solve the inner problem~\eqref{eq:value_function} for $\{p_i^{(0)}\}$; record the objective\nllabel{ln:init}\;
		\For{$m=0,\ldots,M_{\max}-1$}{
			Set incumbent to $\{p_i^{(m)}\},\{r_i^{(t,m)}\}$ and its objective\;
			\ForEach{$i\in\mathcal I$ and direction $s\in\{+1,-1\}$}{
				Form candidate $\tilde p_i = \Pi_{[0,p_i^{\max}]}\!\big(p_i^{(m)}+s\Delta_i\big)$ \nllabel{ln:cand}\;
				Solve inner problem for the candidate prices; evaluate objective \nllabel{ln:solvecand}\;
				\If{feasible and strictly improves incumbent objective}{
					Accept as new incumbent \nllabel{ln:accept}\;
				}
			}
			\If{no candidate accepted}{shrink $\Delta_i\leftarrow\xi\Delta_i,\;\forall i$\nllabel{ln:shrink}}
			Update iterate to accepted incumbent\nllabel{ln:update}\;
			\If{$\|\{p_i^{(m+1)}\}-\{p_i^{(m)}\}\|\le\epsilon_p$ and objective change $\le\epsilon_{\mathrm{obj}}$,
				or $\max_i\Delta_i\le\Delta_{\min}$}{break\nllabel{ln:stop}}
		}
		\Return $\{p_i^\star\},\{r_i^{(t)\star}\}$\;
	\end{algorithm}

	\begin{proposition}
		\label{prop:ao_monotone}
		Every accepted iterate of Algo.~\ref{alg:proposed_ao} is feasible
		for~\eqref{prob:deterministic_operator}, and the objective sequence is
		monotonically nondecreasing.
	\end{proposition}
	
	\begin{proof}
		Feasibility follows because each candidate is evaluated only after
		exact re‑optimization of the reserves under the full set of constraints.
		Line~\ref{ln:accept} in Algo. \ref{alg:proposed_ao} admits only strict objective improvements, hence
		monotonicity. Since prices are bounded by $p_i^{\max}$ and reserves by
		\eqref{eq:request_link_joint_pricing_reserve_operator} and
		\eqref{eq:floor_joint_pricing_reserve_operator}, the monotone objective sequence converges to a finite limit.
	\end{proof}
	The complexity of one exact inner solution over all windows is
	$\mathcal O(|\mathcal T|\,|\mathcal I|\,|\mathcal K|\,B_r B_{\lambda})$,
	where $B_r$ and $B_{\lambda}$ are the bisection iteration counts for the
	reserve and dual variables, respectively. Each outer iteration tests at most
	$2|\mathcal I|$ price candidates, yielding an overall polynomial
	per-iteration complexity of
	$\mathcal O(M_{\max}|\mathcal I|^2|\mathcal T||\mathcal K|B_r B_{\lambda})$, where $M_{\max}$ is the maximum allowed iterations. 
	
	\vspace{-3mm}
	\section{Numerical Analysis}
	\label{sec:numerical}
	
	\subsection{Simulation Setup}
	\label{subsec:setup}
	
	We consider a heterogeneous system with \(I\) tenants and two traffic
	classes \(\mathcal K=\{u,b\}\) per tenant.  The horizon comprises
	\(T=14\) aggregation windows of duration \(\Delta t=1\) minute. The effective
	capacity \(C_{\mathrm{eff}}^{(t)}\) follows a fixed temporal pattern that
	is scaled by a factor \(\rho\) to create different resource regimes, while
	baseline demand \(\bar q_i^{(t)}\) varies across both tenants and time
	windows.  Portfolio weights satisfy \(\alpha_{i,u}\in[0.30,0.70]\) and
	\(\alpha_{i,b}=1-\alpha_{i,u}\).
	The price sensitivity and shortfall penalty weight are linked to the
	portfolio composition via \(\kappa_i=0.5+0.10\,\alpha_{i,u}\) and
	\(w_i=1.0+0.5\,\alpha_{i,u}\), yielding \(\kappa_i\!\in\![0.53,0.57]\)
	and \(w_i\!\in\![1.15,1.35]\).  Table~\ref{tab:sim_params} lists the
	remaining parameters used in the simulation. Finally, all results are obtained from \(N=50{,}000\) independent demand
	realizations and the SLA margin in the deterministic reformulation is
	inflated by \(1.05\) to provide a small conservatism buffer.
	
	\begin{table}[t]
		\centering
		\caption{Simulation parameters}
		\label{tab:sim_params}
		\small
		\begin{tabular}{l c}
			\toprule
			Parameter & Value / Range \\
			\midrule
			Mean consumption ratio, \(\beta_{i,k}\) & \(0.6\)\,--\,\(0.9\) \\
			& (class \& tenant dependent)\\
			Demand standard deviation, \(\sigma_{i,k}\) & \(0.8\)\,--\,\(1.5\) \\
			URLLC SLA target, \(\varepsilon_{i,u}\) & \(10^{-3}\) \\
			eMBB SLA target, \(\varepsilon_{i,b}\) & \(5\times10^{-2}\) \\
			Unit maintenance cost, \(c\) & \(0.8\) \\
			Class severity weights, \((\gamma_u,\gamma_b)\) & \((4,1)\) \\
			Minimum reservation floor, \(\underline r_i\) & \(0.75\,\bar C_{\mathrm{eff}}/I\) \\
			\bottomrule
		\end{tabular}
	\end{table}
	\vspace{-2mm}
	\subsection{Performance of the Proposed Scheme}
	\label{subsec:performance}
	
	To evaluate the performance of the AO-based \textit{Proposed} scheme (Algo.~\ref{alg:proposed_ao}), we compare it
	against five baselines, namely,  \emph{Optimal} that denotes the Gurobi-based solution
	of the deterministic program~\eqref{prob:deterministic_operator},
	\emph{No-Floor}, which removes the minimum reservation requirement
	\(\underline r_i\) while keeping all other constraints unchanged, thus 
	isolating the impact of fairness enforcement,  \emph{Uniform-Price} that 
	applies a common price to all tenants and allocates reserves
	proportionally to the induced requests  subject to capacity and floor
	constraints, \emph{Pricing-Only}, inspired from~\cite{Raftopoulos2024}, which 
	optimizes the prices \(\{p_i\}\) while holding reserves at the minimum
	levels that satisfy reliability and floor conditions, and finally,     \emph{Reserve-Only}
	fixes prices over the horizon and optimizes only the reserve variables
	\(\{r_i^{(t)}\}\) under the same system constraints.
	
	\subsubsection{Profit Performance}
	
	Fig.~\ref{profit} illustrates the profit performance as a function of the number of tenants $I$ and the capacity scaling factor $\rho$. In Fig.~\ref{profit}(left), profit increases nearly linearly with $I$ for all schemes. This behavior follows from the objective in \eqref{prob:deterministic_operator}, where additional tenants contribute additional revenue terms $p_i r_i^{(t)}$ across all windows, while capacity pooling enables more efficient utilization and partial mitigation of shortfall penalties. The \textit{Proposed} scheme closely tracks the \textit{Optimal} benchmark across all system sizes, with a maximum optimality gap of approximately $1.26\%$ at $I=2$, which reduces to below $0.33\%$ for larger $I$. In contrast, baselines exhibit increasing degradation as system size grows. For instance, 
	\textit{Uniform-Price} and \textit{Pricing-Only}, which
	cannot tailor reserves to tenant-specific risk profiles, leave up to
	\(16.6\%\) and $13.6\%$ of potential profit unrealized, respectively. Even \textit{Reserve-Only}, which retains temporal flexibility but abandons price
	differentiation, loses roughly \(4.5\%\). The \textit{No-Floor} variant marginally exceeds the constrained benchmarks because it relaxes the minimum-reservation fairness requirement. These results show that the advantage of the \textit{Proposed} scheme lies in the
	synergy of tenant-differentiated pricing and time-adaptive reserve
	allocation, not in using either mechanism alone.
	
	This near-optimal tracking persists under capacity variation in Fig.~\ref{profit} (right). Profit generally increases with $\rho$ as higher effective capacity allows larger admitted reserves and reduces expected shortfall penalties, thus improving the SInP's net revenue. However, a mild non-monotonic behavior is observed around $\rho=0.9$, which arises from the interaction between time-varying demand and capacity profiles, as discussed in Section~\ref{sec:system}. Despite this effect, the \textit{Proposed} scheme consistently achieves near-optimal performance across all regimes, with a maximum deviation of approximately $1.38\%$ from the \textit{Optimal}. This indicates that AO-based joint optimization effectively adapts pricing and reserve decisions to varying resource availability, enabling the SInP to retain near-optimal profit while controlling reliability-induced penalties.
	
	\begin{figure}[tb]
		\centering
		\includegraphics[scale=0.33]{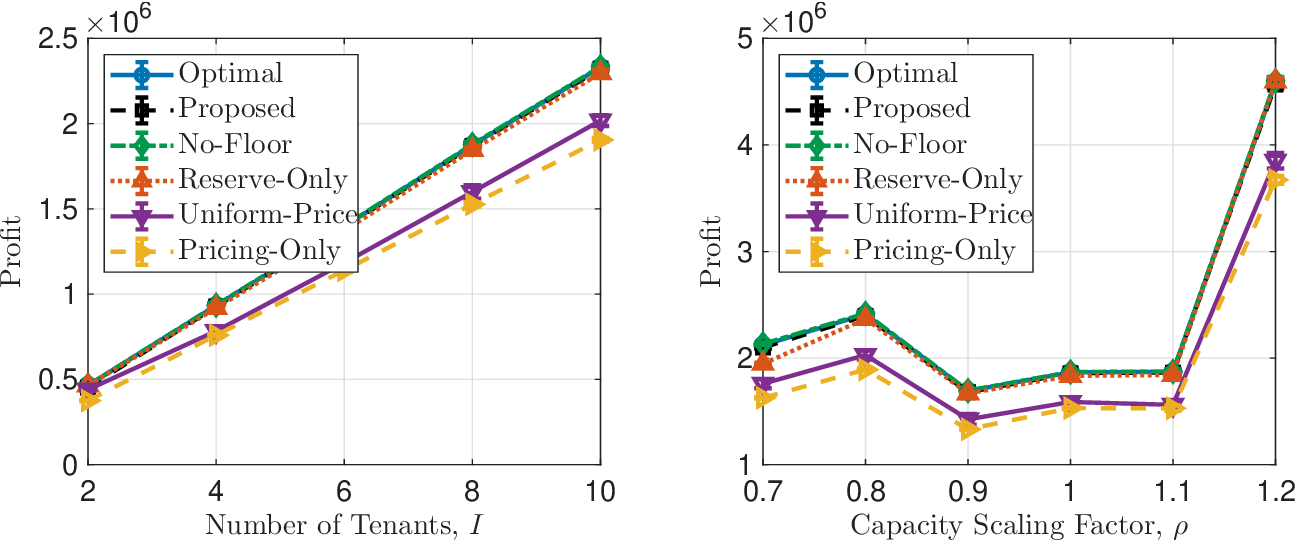}
		\caption{Profit performance versus (left) number of tenants with \(\rho = 1\) and
			(right) capacity scaling factor with $I = 8$. Error bars denote 95\% confidence intervals across independent feasible scenario realizations.}
		\label{profit}
	\end{figure}
	
	\subsubsection{SLA and Utilization Performance}
	
	Fig.~\ref{SLA_Utility} (left) depicts the maximum empirical SLA gap defined as 
	$\max_{i,k,t}\bigl(\hat P_{i,k}^{(t)}-\varepsilon_{i,k}\bigr)$, where nonpositive values indicate that all reliability constraints are satisfied empirically. The \textit{Proposed} scheme maintains a constant gap of approximately $-10^{-3}$ across all capacity regimes, identical to the \textit{Optimal} and \textit{Pricing-Only} schemes. This invariance follows from the deterministic reliability constraints, which remain binding across $\rho$ and guarantee that admitted reserves keep violation probabilities at or below the target levels. Conversely, the \textit{Uniform-Price} scheme violates the SLA for all
	$\rho$, with the maximum empirical gap rising by over $800\%$ across
	the examined capacity regimes,
	since uniform pricing cannot regulate tenant-specific demand under
	relaxed capacity, leading to systematic over-admission. The other
	baselines stay marginally feasible but operate close to the boundary, confirming that joint pricing and reserve decisions are essential for reliability.
	
	Fig.~\ref{SLA_Utility} (right) presents the average utilization performance $\bar U = \frac{1}{T}\sum_{t\in\mathcal T}\frac{\sum_i r_i^{(t)}}{C_{\mathrm{eff}}^{(t)}}$ for varying $\rho$ with fixed $I = 8$. The \textit{Proposed} scheme achieves consistently high utilization between $0.95$ and $0.99$, closely tracking the \textit{Optimal} solution within $2.1\%$. As $\rho$ increases, $\bar U$ dips slightly then stabilizes because higher
	capacity relaxes packing pressure while price-limited demand caps
	reserve growth. \textit{Pricing-Only} maintains a constant utilization of $0.76$, wasting up to $29.8\%$ of capacity due to the absence of adaptive reserve control. The \textit{Reserve-Only} scheme reaches near-full utilization under scarcity but falls to $0.86$ at high $\rho$, reflecting misalignment between fixed pricing and evolving demand. In contrast, \textit{Uniform-Price} exceeds $0.97$ yet incurs severe SLA violations, indicating that high utilization alone does not imply reliable resource allocation. Consequently, the \textit{Proposed} scheme effectively balances high utilization and strict SLA compliance, enabling the SInP to exploit available capacity while avoiding reliability-induced penalties across varying scarcity conditions. In addition, Jain's fairness index remains close to one for all methods, indicating little inter-tenant disparity under the tested conditions.
	\begin{figure}[tb]
		\centering
		\includegraphics[scale=0.33]{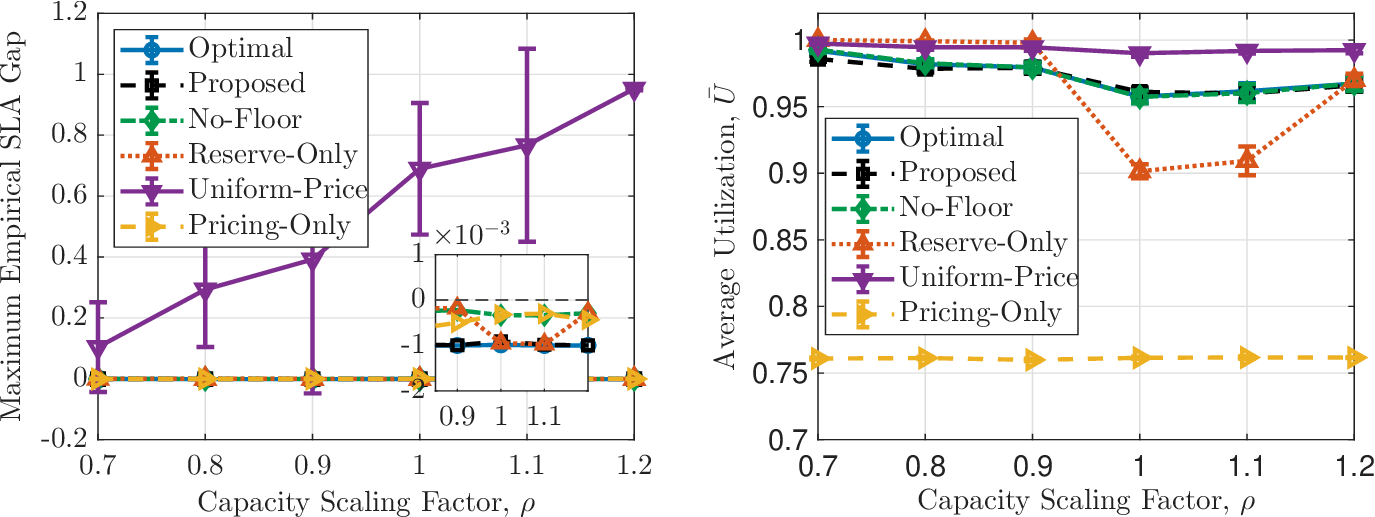}
		\caption{Effect of capacity scaling factors on SLA and utilization with $I = 8$.} 
		\label{SLA_Utility}
	\end{figure}
	
	\subsubsection{Runtime Performance}
	
	Fig.~\ref{runtime} compares the computational runtime of the \textit{Proposed} scheme and the \textit{Optimal} benchmark as the problem size increases in terms of $I$ and $T$ with fixed $\rho = 1$. Compared to \textit{Optimal}, the \textit{Proposed} method exhibits a maximum runtime gap of $22\times$ (more than $95\%$ in runtime reduction). This significant reduction reflects the
	structural decomposition highlighted in Section~\ref{sec:proposed_solution},
	where the AO scheme factorizes the joint optimization into tractable
	subproblems. Minor non-monotonic variation in runtime at larger instances is
	attributable to solver convergence effects rather than algorithmic scaling. In contrast, the \textit{Optimal} benchmark solves the same joint problem via Gurobi’s
	spatial branch-and-bound on a non-convex mixed-integer quadratic program
	with piecewise-linear approximations (i.e., with $M$ breakpoints), leading to a
	worst-case complexity of $\mathcal{O}\!\left((M{-}1)^{IKT} \cdot
	2^{IT} \cdot (IKTM)^{2.5}\right)$. With the near-optimal profit performance in Fig.~\ref{profit} and strict SLA compliance with high utilization in Fig.~\ref{SLA_Utility}, these results prove that the \textit{Proposed} scheme achieves a favorable balance between economic efficiency, reliability, and scalability, making it suitable for practical finite-horizon multi-tenant satellite O-RAN resource orchestration.
	\begin{figure}[tb]
		\centering
		\includegraphics[scale=0.29]{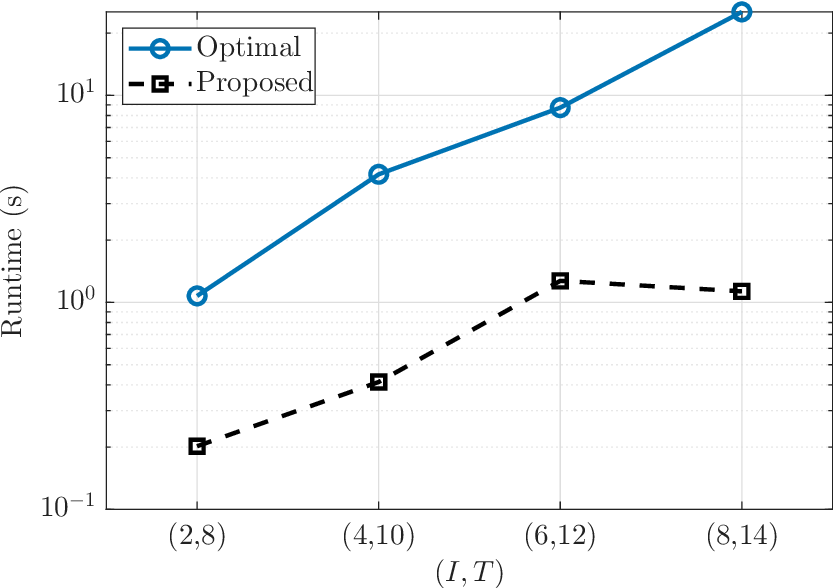}
		\caption{Runtime comparison across different $I$ and $T$.}
		\label{runtime}
	\end{figure}
	\vspace{-2mm}
	\section{Conclusion}
	This work addressed the challenge of coordinating pricing and resource allocation in a satellite O-RAN wholesale environment under stochastic demand and limited time-varying capacity. The problem is formulated to capture the interaction between price-sensitive tenant demand, reliability constraints expressed through SLAs, and shared resource limitations faced by the SInP. 
	The proposed framework that solves this problem integrates pricing and reserve allocation within a unified optimization model, enabling the SInP to jointly manage revenue generation and SLA compliance. Our design leveraged a deterministic representation of probabilistic reliability constraints, which maintains tractability while ensuring robustness against demand variability. Moreover, by structuring our solution through alternating updates, the framework efficiently handled the coupling between decision variables without requiring global optimization. Our numerical results demonstrated consistent near-optimal profit, strict reliability satisfaction, and efficient utilization across varying capacity regimes. Compared with benchmarks that isolate pricing or allocation, our solution avoided systematic resource over-allocation or under-utilization, particularly under scarce resource conditions. 
	Finally, we proved the robust scalability of our solution, thus highlighting its suitability for SInP-side orchestration in shared satellite systems.
	
	
	
	\vspace{-2mm}
	\balance
	\bibliographystyle{IEEEtran}
	\bibliography{Ref_2}
\end{document}